\documentclass[12pt]{amsart}

\usepackage{amsmath}
\usepackage{amsthm}
\usepackage{amssymb}

\usepackage{epstopdf}

\usepackage{fullpage}

\usepackage{geometry}                
\usepackage{graphicx}

\usepackage{mathrsfs}
\usepackage{multicol}

\usepackage{tikz}

\usepackage{url}

\usepackage{pstricks}

\usepackage{url}

\begin{document}

\theoremstyle{plain}
\newtheorem{theorem}{Theorem}
\newtheorem{corollary}[theorem]{Corollary}
\newtheorem{lemma}[theorem]{Lemma}
\newtheorem{claim}[theorem]{Claim}

\theoremstyle{definition}
\newtheorem{definition}[theorem]{Definition}
\newtheorem{example}[theorem]{Example}
\newtheorem{conjecture}[theorem]{Conjecture}

\theoremstyle{remark}
\newtheorem{remark}[theorem]{Remark}

\title{A ternary square-free sequence avoiding factors equivalent to $abcacba$}
\author{James Currie\\
Department of Mathematics and Statistics\\
University of Winnipeg\\
Winnipeg, Manitoba\\
Canada R3B 2E9\\
\texttt{j.currie@uwinnipeg.ca}}
\maketitle
\begin{abstract} We solve a problem of Petrova,
finalizing the classification of letter patterns avoidable by ternary square-free words; we show that there is a ternary square-free word avoiding letter pattern $xyzxzyx$. In fact, we
\begin{itemize}
\item characterize all the (two-way) infinite ternary square-free words avoiding letter pattern $xyzxzyx$
\item characterize the lexicographically least (one-way) infinite ternary square-free word avoiding letter pattern $xyzxzyx$
\item show that the number of ternary square-free words of length $n$ avoiding letter pattern $xyzxzyx$ grows exponentially with $n$.
\end{itemize}.
\end{abstract}
\section{Introduction}
A theme in combinatorics on words is {\bf pattern avoidance}. A word $w$ {\bf encounters} word $p$ if $f(p)$ is a factor of $w$ for some non-erasing morphism $f$. Otherwise $w$ {\bf avoids} $p$. A standard question is whether there are infinitely many words over a given finite alphabet $\Sigma$, none of which encounters a given pattern $p$. Equivalently, one asks whether an $\omega$-word over $\Sigma$ avoids $p$.

The first problems of this sort were studied by Thue \cite{thue,thue2} who showed that there are infinitely many words over $\{a,b,c\}$ which are {\bf square-free} -- i.e., do not encounter $xx$. He also showed that over $\{a,b\}$ there are infinitely many {\bf overlap-free} words -- which simultaneously avoid $xxx$ and $xyxyx$. Thue also introduced a variation on pattern avoidance by asking whether one could simultaneously avoid squares $xx$ and factors from a finite set. For example, Thue showed that infinitely many words over $\{a,b,c\}$ avoid squares, and also have no factors $aba$ or $cbc$.

In combinatorics, once an existence problem has been solved, it is natural to consider stronger questions: characterizations, enumeration problems and extremal problems. Since Thue, progressively stronger questions about pattern-avoiding sequences have been asked and answered: 
\begin{itemize}
\item Gottschalk and Hedlund \cite{gottschalk} characterized the doubly infinite binary words avoiding overlaps.
\item How many square-free words of length $n$ are there over $\{a,b,c\}$? The number of such words was shown to grow exponentially by Brandenburg \cite{brandenburg}.
\item Let ${\bf w}$ be the lexicographically least square-free  $\omega$-word over $\{a,b,c\}$. As the author \cite{allouche} has pointed out,  the method of Shelton \cite{shelton} allows one to test whether a given finite word over $\{a,b,c\}$ is a prefix of ${\bf w}$.
\end{itemize}

Interest in words avoiding patterns continues, and a recent paper by Petrova \cite{petrova} studied {\bf letter pattern avoidance} by ternary square-free words. A word $w$ over $\{1,2,3\}$ avoids the {\bf letter pattern}  $P\in\{x,y,z\}^*$ if no factor of $w$ is an image of $P$ under an injection from $\{x,y,z\}$ to $\{1,2,3\}$. For example, to avoid the letter pattern $xyzxzyx$, a word $w$ cannot contain any of the factors $1231321, 1321231, 2132312, 2312132, 3123213$ and $3213123$. Regarding this particular letter pattern, Petrova remarks at the end of her paper that `(p)roving its avoidance will finalize the classification of letter patterns avoidable by ternary square-free words.'

In this note, we show that there is a ternary square-free word avoiding letter pattern $xyzxzyx$. In fact, we
\begin{itemize}
\item characterize all the (two-way) infinite ternary square-free words avoiding letter pattern $xyzxzyx$
\item characterize the lexicographically least (one-way) infinite ternary square-free word avoiding letter pattern $xyzxzyx$
\item show that the number of ternary square-free words of length $n$ avoiding letter pattern $xyzxzyx$ grows exponentially with $n$.
\end{itemize}
\section{Preliminaries}
We will assume standard notations from combinatorics on words. For reference see the books by Lothaire \cite{loth97,loth02}. In particular, a word is {\bf square-free} if it has no non-empty factor $xx$. 
Let $S=\{1,2,3\}$, $T=\{a,b,c,d\}$ and $U=\{a,c,d\}$. For an alphabet $\Sigma$, we denote by $\Sigma^*$, the set of all finite words over $\Sigma$; by $\Sigma^\omega$, we denote the $\omega$-words over $\Sigma$, which are infinite to the right; by $\Sigma^{\mathbb Z}$ we denote the ${\mathbb Z}$-words over $\Sigma$, which are doubly infinite. Depending on context, a `word' over $\Sigma$ may refer to a finite word, an $\omega$-word or a ${\mathbb Z}$-word.

We put natural orders on alphabets $S$, $T$ and $U$:
$$1 < 2 < 3\mbox{ and }a<b<c<d.$$ These induce lexicographic orders on words over these alphabets; the definition is recursive: if $w$ is a word and $x,y$ are letters, then $wx<wy$ if and only if $x<y$.

Call a word over $S$ {\bf factor-good} if it has no factor of the form $xyzxzyx$ where $\{x,y,z\}=S$; i.e., the factors 1231321, 1321231, 2132312, 2312132, 3123213, 3213123 are forbidden. Call a word over $S$ {\bf good} if it is square-free and factor-good. Petrova's question is whether there are infinitely many good words.

\section{Results on good words}
Theorem~\ref{existence} and Theorem~\ref{struct} below characterize good ${\mathbb Z}$-words. These turn out to be in 2-to-1 correspondence with square-free 
${\mathbb Z}$-words over $U$.

Let $\pi$ be the morphism on $S^*$ generated by $$\pi(1)=1, \pi(2)=3, \pi(3)=2;$$ thus, this morphism $\pi$ relabels 2's as 3's and vice versa.

Let $f$:$T^*\rightarrow S^*$ be the morphism given by 
$$f(a)=1213, f(b)=123, f(c)=1323, f(d)=1232.$$

Let $g$:$U^*\rightarrow T^*$ be the map where $g(u)$ is obtained from a word $u\in\{a,c,d\}^*$ by replacing each factor $ac$ of $u$ by $abc$, each factor $da$ of $u$ by $dba$ and 
each factor $dc$ of $u$ by $dbc$.

\begin{lemma}\label{factor-good} Suppose $u\in U^*$. Then $f(g(u))$ is factor-good.
\end{lemma}
\begin{lemma}\label{square-free} The map $f\circ g:U^*\rightarrow S^*$ is square-free: Suppose $u\in U^*$ is square-free. Then so is $f(g(u))$.
\end{lemma}

\begin{theorem}\label{existence} If ${\bf u}\in U^{\mathbb Z}$ is square-free then $f(g({\bf u}))$is good.
\end{theorem}

\begin{theorem}\label{struct} Let ${\bf w}\in S^{\mathbb Z}$ be good. Exactly one of the following is true:
\begin{enumerate}
\item There is a square-free word ${\bf u}\in U^{\mathbb Z}$ such that ${\bf w}= f(g({\bf u}))$.
\item There is a square-free word ${\bf u}\in U^{\mathbb Z}$ such that ${\bf w}= \pi(f(g({\bf u})))$.
\end{enumerate}
\end{theorem}

We can also characterize the lexicographically least good $\omega$-word:

\begin{theorem}\label{least} The lexicographically least good $\omega$-word is $f(g({\bf u}))$, where ${\bf u}$ is the lexicographically least square-free $\omega$-word over $U$.
\end{theorem}
There are `many' finite good words, in the sense that the number of words grows exponentially with length. For each non-negative integer $n$, let $G(n)$ be the number of good words of length $n$.
\begin{theorem}\label{growth} The number of good words of length $n$ grows exponentially with $n$. In particular, there are positive constants $A$, $B$ and $C>1$ such that $$\sum_{i=0}^n G(i)\ge A+B(C^n).$$
\end{theorem}  

\section{Proof of Theorem~\ref{struct}}
The proof of Theorem~\ref{struct} proceeds via a series of lemmas. Suppose that ${\bf w}\in\Sigma^{\mathbb Z}$ is good. Since ${\bf w}$ is square-free, 
$${\bf w}\in\{12,123,1232,13,132,1323\}^{\mathbb Z}.$$ 
\begin{lemma}\label{1231 or 1321}
Let $w$ be a good word. Then either $|w|_{1231}=0$ or $|w|_{1321}=0$.
\end{lemma}
\begin{proof} If the lemma is false, then either 
\begin{itemize}
\item $w$ contains a finite factor with prefix $1231$ and suffix $1321$ or
\item $w$ contains a factor with prefix $1321$ and suffix $1231$.
\end{itemize}
Without loss of generality up to relabeling, suppose that $w$ contains a factor with prefix 1231 and suffix 1321. Since it is good, $w$ cannot have 1231321 as a factor. Consider then a shortest factor $1231v1321$ of $w$; thus $|1231v1321|_{1231}=1$. 

Exhaustively listing good words $1231u$ with $|1231u|_{1231}=1$, we find that there are only finitely many, and exactly three which are maximal with respect to right extension: 12312131232123, 123132312131232123, 12313231232123. It follows that one of these is a right extension of $1231v1321$; however, none of the three has 1321 as a factor. This is a contradiction.
\end{proof}
Interchanging 2's and 3's if necessary, suppose that ${\bf w}_{1321}=0.$ Thus
$${\bf w}\in\{12,123,1232,13,1323\}^{\mathbb Z}.$$ 

\begin{lemma}\label{4 blocks}
Suppose ${\bf t}\in 1213\{12,123,1232,13, 1323\}^\omega$ is good. Then $${\bf t}\in \{1213,123,1232, 1323\}^\omega.$$
\end{lemma}
\begin{proof} We prove this via a series of claims:\vspace{.1in}

\begin{claim} Neither of 132313 and 21232 is a factor of ${\bf t}$.\end{claim}
\begin{proof}[Proof of Claim] Since ${\bf t}\in1213\{12,123,1232,13,1323\}^{\omega}$, if 132313 is a factor of ${\bf t}$, then so is one of 1323131 and 13231323, both of which end in squares. This is impossible, since ${\bf t}$ is good.
Similarly, if 21232 is a factor of ${\bf t}$, so is one of 121232 and 12321232, both of which begin with squares.\end{proof}

\begin{claim}\label{121} Suppose that $t12uv$ is a factor of ${\bf t}$, where $t$, $u$, $v\in \{12,123,1232,13,1323\}$. Then $u=13$.\end{claim} 
\begin{proof}[Proof of Claim]Word $u$ must be 13 or 1323; otherwise, $12u$ begins with the square 1212. Suppose $u=1323$. By the previous claim, v must have prefix 12. But then $2uv$ has prefix $2132312=xyzxzyx$, where $x=2$, $y=1$, $z=3$; this is impossible. Thus $u=13$.
\end{proof}
\begin{claim}\label{131}Suppose that $tu13v$ is a factor of ${\bf t}$, where $t$, $u$, $v\in \{12,123,1232,13,1323\}$. Then $u=12$.\end{claim} 
\begin{proof}[Proof of Claim] Word $u$ must end with 2; otherwise, $u13v$ contains the square 3131. Thus $u$ must be 12 or 1232. Suppose $u=1232$. By the first claim, $t$ must have suffix 3. But then $tu13$ has suffix $3123213=xyzxzyx$, where $x=3$, $y=1$, $z=2$; this is impossible. Thus $u=12$.
\end{proof}

We have proved that 12 and 13 only appear in ${\bf t}$ in the context 1213. It follows that ${\bf t}\in\{1213,123,1232,1323\}.$
\end{proof}
\begin{corollary}\label{parse w}
Word ${\bf w}\in \{1213,123,1232, 1323\}^{\mathbb Z}.$
\end{corollary}
\begin{proof}
We know that 
${\bf w}\in\{12,123,1232,13,1323\}^{\mathbb Z}.$ If neither of $121$ and $131$ is a factor of ${\bf w}$, then ${\bf w}$ is concatenated from copies of $A=1323$, $B=1232$ and $C=123$. However, $CB$ and $AC1$ contain squares, while $BA12$ contains 2132312, which cannot be a factor of a good word. This implies that $A$, $B$ and $C$  always occur in ${\bf w}$ in the cyclical order $A\rightarrow B\rightarrow C\rightarrow A$, and ${\bf w}$ contains the square $ABCABC$, which is impossible.
We conclude that one of 121 and 131 is a factor of ${\bf w}$. However, as in the proof of Claims~\ref{121} and \ref{131}, factors 12 and 13 can only occur in ${\bf w}$ in the context 1213, so the result follows.
\end{proof}

 By Corollary~\ref{parse w}, $f^{-1}({\bf w})$ exists. Let ${\bf v}\in f^{-1}({\bf w}).$
\begin{lemma} None of $ac$, $aba$, $bd$, $cb$, $da$ and $dc$ is a factor of ${\bf v}$. 
\end{lemma}
\begin{proof}
One checks that $f(ac)$, $f(aba)$, $f(bd)$, $f(cb)1$, $f(da)$ contain squares, and thus cannot be factors of ${\bf w}$. It follows that $ac$, $aba$, $cb$, $da$ and $dc$ are not factors of ${\bf v}$. On the other hand, as in the proof of the previous lemma, $f(d)=1232$ only appears in ${\bf w}$ in the context 123213. It follows that if $cd$ is a factor of ${\bf v}$, then $f(cd)13=1323123213$ is a factor of ${\bf w}$. However, this has the suffix $3123213=xyzxzyx$ where $x=3$, $y=1$, $z=2$. This is impossible.
\end{proof}

\begin{remark} It follows that ${\bf v}$ can be walked on the directed graph ${\mathscr D}$ of Figure~1. 
\end{remark}
\vspace{.1in}
\begin{figure}
\caption{The ${\mathbb Z}$-word ${\bf v}$ can be walked on the directed graph ${\mathscr D}$ below.}
\begin{tikzpicture}[->,shorten >=1pt,auto,node distance=3cm,
  thick,main node/.style={circle,fill=blue!20,draw,font=\sffamily\Large\bfseries},minimum size=1.2cm,every loop/.style={}]

  \node[main node] (1) {$a$};
  \node[main node] (2) [below of=1] {$b$};
  \node[main node] (3) [below right of=2] {$d$};
  \node[main node] (4) [below left of=2] {$c$};

  \path[every node/.style={font=\sffamily\small}]
    (1)  edge  node[left] {} (2)
        edge node[left] {} (3)
    (2)  edge node {} (4)
       edge node[left] {} (1)
    (3)  edge node {} (2)
    (4)  edge node {} (1)
       edge node[left] {} (3);
 
\end{tikzpicture}
\end{figure}
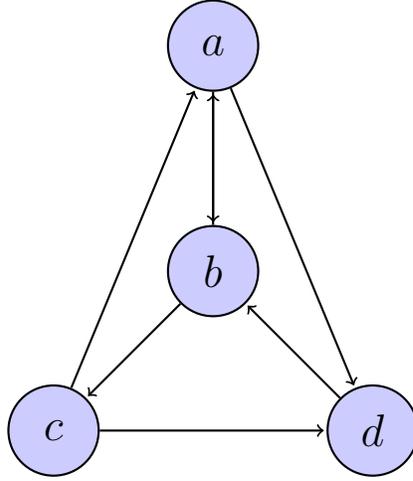

Let $h:\{a,b,c,d\}^*\rightarrow \{a,c,d\}^*$ be the morphism generated by $h(a)=a$, $h(b)=\epsilon$, $h(c)=c$, $h(d)=d$. Thus $h(w)$ is obtained by deleting all occurrences of $b$ in a word $w$. 

Suppose that $w$ is a factor of ${\bf v}$; this implies that $w$ can be walked on the directed graph ${\mathscr D}$. If $w$ does not begin or end with $b$, then
$$w=g(h(w)).$$
Let ${\bf u}=h({\bf v})\in U^{\mathbb Z}$. It follows that ${\bf v}=g({\bf u})$, so that
$${\bf w}=f(g({\bf u})).$$
Word ${\bf u}$ must be square-free; otherwise its image ${\bf w}$ contains a square. This the first alternative in Theorem~\ref{struct} holds.

The other situation occurs if we decide, after Lemma~\ref{1231 or 1321}, that ${\bf w}_{1231}=0.$ As we remarked at that point in our argument, 
this amounts to interchanging 2's and 3's, i.e., applying $\pi$. In such a case, we find that 
$${\bf w}=\pi(f(g({\bf u}))).$$
This completes the proof of Theorem~\ref{struct}.
\section{Proof of Theorem~\ref{existence}}
\begin{proof}[Proof of Lemma~\ref{factor-good}] 
Let $w=f(g(u))$. Each length 7 factor of $w$ is a factor of $f(g(u'))$, some factor $u'\in U^3$. A finite check establishes that $f(u')$ is factor-good for each $u\in U^3.$
\end{proof} 
\begin{proof}[Proof of Lemma~\ref{square-free}] 
Suppose for the sake of getting a contradiction, that $XX$ is a non-empty square in $w=f(v)$. If $|X|\le 2$, then $XX$ is a factor of $f(v')$, some factor $v'$ of $v$ with $|v'|= 2$. However,  we need only consider 

$$v'\in \{ab, ad, ba, bc, ca, cd, db\}.$$ 

(We can walk $v'$ on ${\mathscr D}$.) In each case, we check that $f(v')$ is square-free. From now on, then, suppose that $|X|\ge 3$; in this case we can write 

$$XX=qf(v_1v_2\cdots v_{n-1})p=qf(v_{n+1}v_{n+2}\cdots v_{2n-1})p,$$

  where $v_0v_1\cdots v_{n-1}v_nv_{n+1}v_{n+2}\cdots v_{2n-1} v_{2n}$ is a factor of $v$, $q$ is a suffix of $f(v_0)$, $p$ is a prefix of $f(v_{2n})$, $f(v_n)=pq$, and the $v_i\in T$. It follows that $v_i=v_{n+i}$, $1\le i\le n-1$.

If $v_0=v_n$, then $v$ contains the square $(v_0v_1v_2\cdots v_{n-1})^2$; similarly, if $v_n=v_{2n}$, then $v$ contains the square $(v_1v_2v_3\cdots v_{n})^2$. Since $v$ is square-free, we deduce that $v_n\ne v_0,v_{2n}$. From the condition that $f(v_n)$ is concatenated from a prefix of $v_{2n}$ and a suffix of $v_0$, where $v_n\ne v_0,v_{2n}$, we deduce that $v_n=b$.

From the definition of $g$ and the fact that $v_n=b$, it follows that $v_{n-1}v_nv_1\in\{abc, dba,dbc\}.$ If $v_{n-1}=d$, the definition of $g$ would force $v_n=v_{2n}=b$, contradicting $v_{2n}\ne v_n$. We conclude that 
 $v_{n-1}v_nv_1=abc.$
However, if $v_1=c$, the definition of $g$ forces $v_n=v_{0}=b$, contradicting $v_{0}\ne v_n$. 
\end{proof}
\section{Proof of Theorem~\ref{least}}
Let ${\bf u}$ be the lexicographically least square-free $\omega$-word over $U=\{a,c,d\}$, and let ${\bf t}=f(g({\bf u}))$. It follows that ${\bf u}$ has prefix $ac$, so that ${\bf t}$ has prefix
$p=f(g(ac))=f(abc)=12131231323$. A finite search shows that $p$ is the lexicographically least good word of length 11. It will therefore suffice to show that ${\bf t}$ is the lexicographically least good $\omega$-word with prefix $p$. 

Suppose that ${\bf t}_1$ is a good $\omega$-word with prefix $p$. By Lemma~\ref{1231 or 1321}, it follows that
$|{\bf t}_1|_{1321}=0$, and from the proof of Theorem~\ref{struct}, we conclude that ${\bf t}_1=f(g({\bf u_1}))$, for some square-free word ${\bf u_1}$. It remains to show that ${\bf u_1}$ is lexicographically greater than or equal to ${\bf u}$. Suppose not.

Since ${\bf t_1}$ has prefix $p$, word $ac$ must be a prefix of ${\bf u_1}$, and ${\bf u}$, ${\bf u_1}$ agree on a prefix of length at least 2. Let $qrs$ and $qrt$ be prefixes of ${\bf u_1}$ and ${\bf u}$, respectively, where $r,s,t\in\{a,c,d\}$, and $s$ is lexicographically less than $t$.
\begin{itemize}
\item If $r=a$, then we cannot have $s=a$, since ${\bf u_1}$ is square-free. We therefore must have $s=c$ and $t=d$. It follows that ${\bf t_1}$ has prefix $f(g(qa)bc)=f(g(qa))1231323$, and ${\bf t}$ has prefix $f(g(qa)d)=f(g(qa))1232$, and we see that ${\bf t_1}$ is lexicographically less than ${\bf t}$. This contradicts the minimality of ${\bf t}$.
\item If $r=c$, then we must have $s=a$ and $t=d$. It follows that ${\bf t_1}$ has prefix $f(g(qca))=f(g(qc))1213$, and ${\bf t}$ has prefix $f(g(qcd))=f(g(qc))1232$, and again ${\bf t_1}$ is lexicographically less than ${\bf t}$, giving a contradiction.
\item If $r=d$, then we must have $s=a$ and $t=c$. It follows that ${\bf t_1}$ has prefix $f(g(qd)ba)=f(g(qd))1231213$, and ${\bf t}$ has prefix $f(g(qd)bc)=f(g(qc))1231323$, and again ${\bf t_1}$ is lexicographically less than ${\bf t}$.
\end{itemize}
We conclude that ${\bf u_1}$ is lexicographically greater than or equal to ${\bf u}$, and ${\bf u}$ is the lexicographically least square-free $\omega$-word over $U$, as claimed.
\section{Proof of Theorem~\ref{growth}}
Let $C(n)$ be the number of length $n$ square-free words over $U$. As shown by Brandenburg \cite{brandenburg}, for $n>2$, $C(n)\ge 6\left(2^{n\over 21}\right)$.
The map $f\circ g$ is injective. Since $g$ simply adds $b$'s between some pairs of letters, $|u|\le |g(u)|< 2|u|$; also, $3|u|\le |f(u)|\le 4|u|$. Let $u\in U^*$ be square-free. By the Lemmas~\ref{factor-good} and \ref{square-free}, $f(g(u))$ is good. Also,
$3|u|\le |f(g(u))|< 8|u|$. We deduce that distinct square-free words over $U$ of lengths between 3 and $(n+1)/8$ correspond to distinct good words of lengths between 9 and $n$.
It follows that
$$\sum_{i=3}^{\lfloor (n+1)/8\rfloor}6\left(2^{n\over 22}\right)\le \sum_{i=9}^n G(i),$$
and the theorem follows with $A=\sum_{i=0}^8 G(i)$, $B=6$ and $C=2^{1\over 22}$.

\end{document}